\documentclass[12pt]{article}
\usepackage{enumitem}
\usepackage{amsmath}
\usepackage{amsthm}
\usepackage{fancyhdr}
\usepackage{mathrsfs}
\usepackage{natbib}
\usepackage[margin=1in]{geometry}
\usepackage{color}
\usepackage{multirow,array}
\usepackage{amsmath,amsthm,amssymb}
{
\theoremstyle{plain}
\newtheorem{theorem}{Theorem}[section] \newtheorem{proposition}{Proposition}[section] 

\newtheorem{definition}{Definition}

\newtheorem{example}{Example}
\newtheorem{algorithm}{Algorithm}
  }
\usepackage[colorlinks = true,
            linkcolor = blue,
            urlcolor  = blue,
            citecolor = blue,
            anchorcolor = blue]{hyperref}
\usepackage{hyperref}
\usepackage{verbatim}
\usepackage{tikz}
\usepackage{mathtools}

\usepackage{appendix}
\usepackage{float}
\usepackage{graphicx}
\usepackage{caption}
\usepackage{subcaption}
\usepackage[utf8]{inputenc}
\usepackage[T1]{fontenc}    
\usepackage{hyperref}       
\usepackage{url}            
\usepackage{booktabs}       
\usepackage{amsfonts}       
\usepackage{nicefrac}       
\usepackage{microtype}      

\usepackage{dsfont}
\usepackage{indentfirst}

\usepackage{adjustbox}

\numberwithin{equation}{section}

\title{Using Aggregate Relational Data to Infer Social Networks
\thanks{I am grateful to Hiroaki Kaido, Ivan Fernandez-Val, and Jean-Jacques Forneron for their guidance and support. I also thank participants of Boston University econometrics workshop and BU-BC Greenline workshop in Econometrics for their helpful comments.}
}
\author{Xunkang Tian\footnote{European Research University. Email: \href{mailto:xunkang.tian@eruni.org}{xunkang.tian@eruni.org}}}
\date{\today}

\begin{document}

\maketitle

\begin{abstract}
This study introduces a novel approach for inferring social network structures using Aggregate Relational Data (ARD), addressing the challenge of limited detailed network data availability. By integrating ARD with variational approximation methods, we provide a computationally efficient and cost-effective solution for network analysis. Our methodology demonstrates the potential of ARD to offer insightful approximations of network dynamics, as evidenced by Monte Carlo Simulations. This paper not only showcases the utility of ARD in social network inference but also opens avenues for future research in enhancing estimation precision and exploring diverse network datasets. Through this work, we contribute to the field of network analysis by offering an alternative strategy for understanding complex social networks with constrained data.
\end{abstract}

\begin{center}
\end{center}

\clearpage

\section{Introduction}

The analysis of network structures is increasingly recognized as a critical tool in economic and social sciences, providing invaluable insights into the complex interplay of relationships that underpin various economic phenomena and social interactions. The detailed mapping of these networks, capturing the myriad connections between individuals and their attributes, is pivotal for an understanding of collective behaviors and market dynamics. However, the meticulous collection and processing of network data pose substantial challenges. Traditional methodologies necessitate conducting comprehensive surveys to capture every potential link among individuals, followed by a labor-intensive process of collating, matching, and structuring this data into analyzable formats. This not only demands significant resources but also introduces logistical hurdles that can hamper the scope and scale of research endeavors.

Recognizing these challenges, \cite{breza2020using} advocates for a paradigm shift towards leveraging Aggregate Relational Data (ARD) as a pragmatic alternative. Unlike traditional network data that requires detailed enumeration of individual-level links, ARD focuses on summarizing network characteristics at an aggregate level, such as the total number of connections an individual has or the presence of specific attributes within one's network. This approach, while streamlining data collection through simpler questionnaire-based methods, inherently trades off the granularity of information for operational efficiency and reduced complexity.

This paper seeks to bridge the gap between the rich insights afforded by detailed network data and the practical advantages of ARD. Drawing on the theoretical foundation laid by \cite{mele2017structural}, it extends their inferential framework to incorporate ARD, thereby tackling the dual challenges of computational complexity and data accessibility. By employing a Markov chain Monte Carlo (MCMC) method in tandem with the variational approximation techniques by \cite{melezhu2017approximate}, the proposed methodology yields a sequence of parameter estimates that approximate the posterior distribution of interest. This approach acknowledges the inherent limitations of ARD—primarily, its reduced detail level—and attempts to compensate for these through strategic selection of ARD statistics that retain critical information relevant to understanding network formation and evolution.

To rigorously assess the viability and effectiveness of ARD in network analysis, the study conducts a Monte Carlo Simulation, designed to simulate networks based on predefined parameters. This simulation not only tests the robustness of the proposed estimation method but also provides empirical evidence on the trade-offs involved in adopting ARD over traditional network data. The findings from this exercise are expected to shed light on the conditions under which ARD can serve as a reliable substitute, highlighting the potential limitations and the strategies for mitigating information loss.

In essence, this paper contributes to the evolving landscape of network analysis by advocating for the strategic use of ARD in situations where traditional data collection methods are impractical. Through detailed simulations and analyses, this research underscores the potential of ARD to unlock new avenues for economic and social network studies, enriching the understanding of the underlying mechanisms that drive network formation and dynamics.

\section{Representation of Aggregate Relational Data}\label{section:ard}

Aggregate relational data offers a different perspective when detailed network data is inaccessible. Unlike conventional methods where the entire network data is required, ARD focuses on capturing certain relational features of the network. Typically, ARD is sourced from survey queries such as "how many friends do you have?", "Among your friends, how many smoke?", or "How many of your friends possess a college degree?". These questions shed light on specific features of individuals within the larger network context.

A notable advantage of ARD is its cost-effectiveness, as collecting comprehensive network data can be prohibitively expensive and challenging. The standard approach necessitates an exhaustive census, extensive surveys to every individual, and a matching mechanism correlating responses with social ties. Field surveys can make these steps particularly burdensome. By shifting from complete network data to ARD, researchers sacrifice granularity for efficiency, but the resultant ARD still reveals significant information about network structures and aids in estimating unknown coefficients \citep{breza2020using}.

To express the ARD, I first introduce some basic notations. I consider a finite population \(I = \{1, 2, ..., n\}\) of individuals. 
A network is represented as \(g  \in \mathcal{G}\), where \(\mathcal{G}\) denotes the set of all possible networks. The characteristics of individual \(i\) are denoted by \(X_i\), and the combination of all individuals' characteristics is represented as \(X = (X_1, ..., X_n)^T \in \mathcal{X}\), where \(\mathcal{X}\) is the set of all conceivable characteristics. The utility of individual \(i\) within a network structure \(g\), given the observed characteristics \(X\), is encapsulated by \(U_i(g, X, \varepsilon; \theta)\). I should note that this utility function is only known up to a finite-dimensional parameter \(\theta \in \Theta\), where \(\Theta \subseteq \mathds{R}^{d_\theta}\) denotes the domain and \(d_\theta\) the dimension of the parameter space. As for \(\varepsilon \in \mathcal{E}\), it signifies the unobserved preference shock, which is assumed to follow the i.i.d. Type I extreme value distribution.

For ARD-based survey questions, it is imperative that all features gleaned from the questions are encapsulated within the characteristics \(X\). These questions should be linked to the network's structure, and independent of the parameters of interest, \(\theta\). Formally, ARD statistics can be described by a mapping function \(\psi\), which transforms the unobserved network \(g\) and observed characteristics \(X\) into a \(d_{\psi}\)-dimensional real vector: \(\psi: \mathcal{G}\times \mathcal{X} \rightarrow R^{d_{\psi}}\). 

When researchers observe the aggregate relational data, denoted by $\psi_0$, the latent network is unobservable, while it is assumed to satisfy the equilibrium condision specified in \cite{mele2017structural}. 
The parameter of interest is \(\theta\). The primary goal of our study is to infer the value of \(\theta\) after observing \(\psi_0\) and \(X\). The estimation of \(\theta\) will provide insights into the underlying structure and individuals' preferences.

\begin{example}
This example is based on \cite{banerjee2017credit}. Their data collection approach leveraged a questionnaire comprised of the following questions, designed to probe the relational dynamics within a neighborhood:
\begin{enumerate}
    \item How many households do you know where a woman has ever given birth to twins?
    \item How many households do you know that have a permanent government employee?
    \item How many households do you know with 5 or more children?
    \item How many households do you know where any child has studied beyond the 10th standard?
    \item How many households do you know where any adult contracted typhoid, malaria, or cholera in the past six months?
    \item How many households do you know where an adult has been apprehended by the police?
    \item How many households do you know where at least one woman has had a subsequent marriage?
    \item How many households do you know where a man is concurrently married to more than one wife?
\end{enumerate}

In this context, \(n\) denotes the total count of respondents. The ARD statistics dimension is represented by \(d_{\psi}=8n\). Specifically, elements ranging from \(8i-7\) to \(8i\) in \(\psi(g,X)\) encapsulate the answers provided by the \(i^{th}\) respondent to the aforementioned questions.
\end{example}

The utility of player \(i\) from a network \(g\) with population attributes \(X = (X_1,...,X_n)\) and parameter \(\theta= (\theta_u,\theta_m,\theta_v,\theta_w)\) is given by
\begin{equation}\label{eq:utilityfunction}
    U_i(g,X,\varepsilon;\theta)=
    \sum_{j\neq i} g_{ij} (u_{ij}^{\theta_u}+\varepsilon_{ij})
    + \sum_{j\neq i} g_{ij}g_{ji} m_{ij}^{\theta_m}
    + \sum_{j\neq i} g_{ij} \sum_{k\neq i,j} g_{jk} v_{ik}^{\theta_v}
    + \sum_{j\neq i} g_{ij} \sum_{k\neq i,j} g_{ki} w_{kj}^{\theta_w},
\end{equation}
where
\begin{itemize}
    \item \(u_{ij}^{\theta_u}\equiv u(X_i,X_j;\theta_u)\) is the direct net utility from the link to individual \(j\), with \(\varepsilon=\{ \varepsilon_{ij} \}_{i,j\in I, \; i\neq j}\) being the unobserved direct shock;
    \item \(m_{ij}^{\theta_m}\equiv m(X_i,X_j;\theta_m)\) captures the additional utility if the link is reciprocal (i.e., $g_{ij}=g_{ji}=1$);
    \item \(v_{ik}^{\theta_v}\equiv v(X_i,X_k;\theta_v)\) represents the utility of indirect connection to individual \(k\). When individual \(i\) is deciding whether to create a link to \(j\), he considers \(j\)'s connections and the utility from the indirect link;
    \item \(w_{kj}^{\theta_w}\equiv w(X_k,X_j;\theta_w)\) corresponds to a popularity effect. When individual \(i\) forms a link to \(j\), he inadvertently creates an indirect link to every individual \(k\) who already has a link to \(i\), generating an externality for them.
\end{itemize}
All four terms are real-valued functions. I assume 
\begin{eqnarray*}
  m(X_i,X_j;\theta_m)=m(X_j,X_i;\theta_m) \quad \text{for all } i,j\in I, \\
  w(X_k,X_j;\theta_v)=v(X_k,X_j;\theta_v) \quad \text{for all } k,j\in I.
\end{eqnarray*}
The following potential function summarizes the players' incentives:
 \citep{mele2017structural}: 
\begin{equation}\label{eq:potentialfunction}
  Q(g,X,\varepsilon;\theta)=
  \sum_{i=1}^n \sum_{j\neq i} g_{ij} (u_{ij}^{\theta_u}+\varepsilon_{ij})
  +\sum_{i=1}^n \sum_{j\neq i} g_{ij}g_{ji} m_{ij}^{\theta_m}
  +\sum_{i=1}^n \sum_{j\neq i} \sum_{k\neq i,j} g_{ij}g_{jk} v_{ik}^{\theta_v}.
\end{equation}

Write $\bar{Q}(g,X;\theta) \equiv  \left. Q(g,X,\varepsilon;\theta)\right|_{\varepsilon=0}$ for convenience. 
The stationary distribution of networks in the network formation model  is represented by: 
\begin{equation}\label{stationarydistri}
  \pi(g|X;\theta)=\frac{\exp[\bar Q(g,X;\theta)]}{\sum_{\omega\in \mathcal{G}}\exp[\bar Q(\omega,X;\theta)]}
\end{equation}

With the  stationary distribution of network formation model in place, the likelihood of observing a particular ARD \(\psi_0\) can be denoted as:
\begin{equation}\label{likeliauxi}
  L(\psi_0,X;\theta) = \sum_{g\in \mathcal{G}} 1\{ \psi(g,X)=\psi_0 \} \pi(g|X;\theta)      
\end{equation}
Building on this, the posterior distribution of \(\theta\) with a predefined prior distribution \(p_0(\cdot)\) in light of the ARD observation \(\psi_0\) is expressed as:
\begin{equation}\label{posteriordistribution}
  p(\theta|\psi_0,X) =  \frac{ p_0(\theta) \sum\limits_{g\in \mathcal{G}}  1\{ \psi(g,X)=\psi_0 \}   \pi(g|X;\theta) }
  { \int_{\Theta} p_0(\vartheta) \sum\limits_{g\in \mathcal{G}}  1\{ \psi(g,X)=\psi_0 \}   \pi(g|X;\vartheta) d\vartheta  }
\end{equation}

Aggregate Relational Data  statistics act as intermediaries for network representation, thereby necessitating the encapsulation of key network features. 
Recall that if a entire network $g_0 \in \mathcal G$ is observed, the posterior distribution for $\theta$ is given by 
\begin{equation}\label{postfull}
  p(\theta|g_0,X) = \frac{p_0(\theta) \pi(g_0|X;\theta) }
  { \int_{\Theta} p_0(\vartheta) \pi(g_0|X;\vartheta) d\vartheta }.
\end{equation}
where $p_0(\cdot)$ is a prior distribution over $\Theta$.  
While the posterior distributions from observing the entire network \eqref{postfull} and ARD \eqref{posteriordistribution} differ, I discuss the conditions ensuring their close resemblance in the following section.

\subsection{Selection of Aggregate Relational Data}\label{sec:appendix:ard}

Functionally, an ARD statistic partitions the network space. Specifically, networks with identical ARD values are categorized within the same set, as detailed below.

\begin{definition}
Given a collection of networks $\mathcal{G}$, and any ARD statistic denoted by $\psi: \mathcal{G} \times \mathcal{X} \rightarrow \mathds{R}^{d_{\psi}}$, the domain of $\psi$ is defined as $\tau_{\psi,X}=\left\{ t \mid \exists g \in \mathcal{G} \text{ such that } t=\psi(g,X) \right\}$. This statistic subsequently segments the space $\mathcal{G}$ into partition sets $A_t=\left\{ g\in \mathcal{G} \mid \psi(g,X)=t \right\}$ that are both mutually exclusive and collectively exhaustive, for each $t \in \tau_{\psi,X}$.
\end{definition}

Such partitioning inherently results in information loss by undifferentiating networks within the same set. An optimal ARD statistic would retain vital information regarding parameters of interest, denoted by $\theta$. This sentiment aligns with the definition of sufficiency, as outlined by \citep{casella2021statistical}.

\begin{definition}
For given characteristics $X$, an ARD statistic $\psi$ achieves sufficiency for parameter $\theta$ if and only if the conditional probability $\Pr[g|\psi(g,X);\theta]$ is independent of $\theta$, applicable across all $g \in \mathcal{G}$ and $\theta \in \Theta$.
\end{definition}

To elucidate this concept, we provide an illustrative example below.

\begin{example}[Sufficient ARD Statistic for a Simple Network]\label{examp:suffi}
Consider a network constituted by $n=4$ individuals: A, B, C, and D, with respective wealth values $W$ of 600, 500, 200, and 100. Their utility function is represented as:
\begin{equation*}
    U_i(g,W,\varepsilon;\theta)= \sum_{j=1}^n g_{ij} \left[\theta_0+\theta_1 1\left\{ |W_i-W_j| = \min_l |W_i-W_l| \right\} + \varepsilon \right]
\end{equation*}
Here, the aspects of indirect friendships and reciprocal friendships are disregarded. We focus on the parameter $\theta=(\theta_0,\theta_1)$, with the stipulation that $\theta_1>\theta_0$.

\begin{figure}[htbp]
    \centering
    \includegraphics[width=0.2\linewidth]{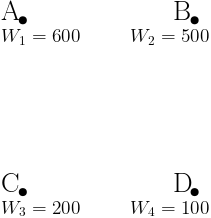}
\caption{Illustrative Representation of a Simple Network}
\label{wealth}
\end{figure}

The ARD statistics can be generated based on these inquiries:
\begin{enumerate}
    \item How many outward links do you possess?
    \item How many inward links are directed to you?
    \item How many links lead from you to individuals with wealth exceeding 400?
\end{enumerate}

Using the aforementioned inquiries, we can capture key relational dynamics within the network. It is feasible to ascertain the sufficiency of these metrics in relation to our parameters of interest. 
\end{example}

\begin{proposition}\label{prop:suffiexamp}
  The ARD statistics delineated in Example \ref{examp:suffi} are sufficient for estimating the parameters of interest $\theta$.
\end{proposition}

Sufficiency, in the context of our model, dictates that the probability distribution of the network remains independent of $\theta$, conditional on the ARD statistics. This independence entails that the ARD statistics harness every ounce of information on $\theta$ embedded within the network. Integrating this principle with the network formation model culminates in the subsequent theorem:

\begin{theorem}\label{sufficientaux}
When $\psi$ proves sufficient for $\theta$, the posterior distribution drawn from observing $\psi(g_0,X)$ in Equation \eqref{posteriordistribution} aligns perfectly with the posterior distribution sourced from observing $g_0$ in Equation \eqref{postfull}.
\end{theorem}

For a fixed network $g_0$, the conditional probability of $g_0$ with respect to its ARD $\psi(g_0,X)$ can be depicted as: 
\begin{equation*}
  P[g_0|\psi(g_0,X);\theta]=\frac{\pi(g_0|X;\theta)}{L(\psi(g_0,X),X;\theta)}
\end{equation*}

Given that a sufficient ARD statistic can impeccably replicate the posterior distribution, we conduct a closer examination. To elucidate, consider the following decomposition: 
\begin{equation*}
  \frac{1}{P[g_0|\psi(g_0,X);\theta]}  
   = \sum_{g\in\mathcal{G}} 1\{ \psi(g,X)=\psi(g_0,X) \} \exp\left[ \bar Q(g,X;\theta)-\bar Q(g_0,X;\theta) \right]
\end{equation*}

The decomposition reveals a pivotal characteristic: for sufficiency of $\psi$, within each partition, regardless of the choice of $g_0$, the summation of the exponential difference between the potential function of $g_0$ and all other networks $g$ must remain uninfluenced by $\theta$. 
This insight imparts several nuanced takeaways: 
\begin{itemize}
    \item Even when identifying or collecting a sufficient statistic proves challenging, opting for an insufficient statistic wherein $P[g_0|\psi(g_0,X);\theta]$ showcases relative steadiness against the fluctuation of $\theta$ can more closely emulate the target posterior distribution \eqref{postfull}.
    \item A higher dimensional ARD statistic, denoted $d_{\psi}$, emerges as a preferable choice. A more refined partition curtails the variance in $P[g_0|\psi(g_0,X);\theta]$.
    \item Incorporating more details from the network formation model into the ARD statistic, $\psi$, can further reduce the variance of $P[g_0|\psi(g_0,X);\theta]$.
\end{itemize}

\subsection{Estimation with Aggregate Relational Data}

To draw inferences from the posterior distribution in \eqref{posteriordistribution}, we encounter the predicament where the network \( g_0 \) is obscured, leaving only the ARD statistics accessible. The computation now becomes considerably daunting. 

As in \cite{mele2017structural}, the normalizing constant $c(X;\theta)\equiv \sum_{\omega\in \mathcal{G}}\exp[\bar Q(\omega,X;\theta)]$ brings challenges in computation within the MCMC framework. 
Inherent in its computation is the exhaustive exploration of all potential networks in $\mathcal{G}$. Given a sample size of $n$ nodes, the magnitude of possibilities is overwhelming: $|\mathcal{G}|=2^{n(n-1)/2}$. The exponential relationship between the number of possible configurations and $n$ renders the computation of the posterior distribution infeasible once $n>10$. 
However, in the case of ARD, 
not only is the normalizing constant \( c(X;\theta) \) intractable, but the summation across \( g \) in \eqref{likeliauxi} also becomes problematic. This is because it requires an exhaustive traversal of every possible network \( g \) to determine the adherence to the ARD condition \( \psi(g,X) = \psi_0 \).

Due to these computational intricacies, traditional methodologies such as the general Metropolis-Hasting method or the exchange algorithm falter. They fail to yield a reasonable acceptance rate without computationally expensive traversals of \( g \).

To navigate through the challenges of the ARD scenario, we posit the underlying network as a latent variable. Consequently, we incorporate it into our algorithmic framework. Assuming \( p_c(\theta,g) \) denotes the joint prior of \( (\theta,g) \), the posterior distribution of \( (\theta,g) \) post observation of \( \psi_0 \) can be rearticulated as: 
\begin{equation} \label{posteriortwovar}
  p(\theta,g|\psi_0,X)=
  \frac{p_c(\theta,g)  \mathds{1}\left\{ \psi(g,X)=\psi_0 \right\} \pi(g|X;\theta) }
  {\mathcal{Z}(\psi_0,X)},
\end{equation}
where the normalizing term is defined as: 
\begin{equation*}
  \mathcal{Z}(\psi_0,X)= \int_{\Theta} \sum_{\omega\in\mathcal{G}} p_c(\vartheta,\omega) \mathds{1}\left\{ \psi(\omega,X)=\psi_0 \right\} \pi(\omega|X;\vartheta)  d\vartheta.
\end{equation*}

A pivotal criterion ensuring that \eqref{posteriordistribution} mirrors the marginal distribution of \eqref{posteriortwovar}, concerning summation over network \( g \), is that the prior ratio, \( \frac{p_c(\theta, g)}{p_0(\theta)} \), remains consistent and uninfluenced by both \( \theta \) and \( g \). In real-world implementations, the formation of \( p_c(\theta,g) \) from \( p_0(\theta) \) to accommodate ARD estimations is typically achieved by postulating a uniform distribution of \( g \) spanning all conceivable network architectures relative to \( p_0(\theta) \).

While the numerator of \eqref{posteriortwovar} eliminates the need for summation over \( g \in \mathcal{G} \), it presents computational challenges when estimating both \( \theta \) and \( g \) in the MCMC framework. Specifically, the persistent issue of the intractable normalizing constant \( c(X;\theta) \) complicates the sampling procedure for both \( \theta \) and \( g \). 
\footnote{In an attempt to mitigate this issue, one could consider applying the exchange algorithm twice: first for \( g \) and then for \( \psi \). However, this approach is problematic for several reasons. The most crucial issue is meeting the detailed balance condition for the algorithm to converge. This condition necessitates that the potential function \( \bar Q(g,X;\vartheta) \) satisfy:
\begin{equation*}
\bar Q(g,X;\vartheta) + \bar Q(g',X;\vartheta') = \bar Q(g',X;\vartheta) + \bar Q(g,X;\vartheta')
\end{equation*}
for all \( \theta, \theta' \in \Theta \) and \( g,g' \in \mathcal{G} \). 
In our context, this equality is implausible because the potential function encapsulates the sum of all player incentives for a given network state. Swapping \( g \) and \( g' \) would unavoidably change the summative value of the function, breaking the required symmetry and violating the detailed balance condition. Thus, iterative application of the exchange algorithm is not feasible for our problem.}

\subsubsection{Variational Approximation}

Addressing the computational challenges posed by the intractable normalizing constant \( c(X;\theta) \) calls for innovative solutions. Once we successfully approximate \( c(X;\theta) \), we can then employ the Metropolis-Hasting technique to sample both \( \theta \) and \( g \) from the posterior, as detailed in \eqref{posteriortwovar}.

An effective method for handling the intractable normalizing constant has been proposed by \cite{melezhu2017approximate}. This method employs variational approximation \citep{jordan1999introduction} to compute the elusive constant. Central to this approach is the identification of an approximate likelihood, a function of the network denoted as \( q(\cdot) \), which minimizes the Kullback-Leibler divergence from the true likelihood \( \pi(\cdot|X;\theta) \) for a specific realization of \( X \) and a given \( \theta \). It is noteworthy that, given the dynamics of network formation, this likelihood is inherently defined by the stationary distribution \eqref{stationarydistri}.

Recalling that the normalizing constant is essentially the sum of exponentiated potential functions, let us begin by defining the statistics \( T \) and \( \phi \) as:
\begin{align*}
  T(g,X;\theta) &\equiv  \frac{1}{n^2} \bar Q(g,X;\theta),    \\
  \phi(X;\theta) &\equiv  \frac{1}{n^2} \log \sum_{\omega\in\mathcal G}\exp [\bar Q(\omega,X;\theta)].
\end{align*}
Consequently, the true likelihood can be represented by:
\begin{equation*}
  \pi(g|X;\theta) = \exp\left\{ n^2[ T(g,X;\theta) - \phi(X;\theta) ] \right\}.
\end{equation*}

The Kullback-Leibler divergence between \( q(\cdot) \) and \( \pi(\cdot|X;\theta) \) is then given by:
\begin{align*}
  KL(q\|\pi) &= \sum_{\omega\in\mathcal{G}} q(\omega) \log \left[ \frac{q(\omega)}{\pi(\omega|X;\theta)} \right], \\
  &= \sum_{\omega\in\mathcal{G}} q(\omega)\left[ \log q(\omega) - n^2 T(\omega,X;\theta) + n^2 \phi(X;\theta)  \right].
\end{align*}
Since \( KL(q\|\pi) \geq 0 \) inherently, rearranging this inequality to isolate \( q \), we obtain:
\begin{equation}\label{ineqpsi}
  \phi(X;\theta) \geq \mathbb{E}_q[T(g,X;\theta)] + \frac{1}{n^2}\mathcal{H}(q),
\end{equation}
where \( \mathcal{H}(q) = - \sum_{\omega\in\mathcal{G}} q(\omega) \log q(\omega) \) denotes the entropy of the distribution \( q \), and \( \mathbb{E}_q[T(g,X;\theta)] \) signifies the expected value of the re-scaled potential function under distribution \( q \).

Notably, \( \phi(X;\theta) \) acts as a re-scaled normalizing constant. From Equation \eqref{ineqpsi}, the right-hand side offers a lower bound. Given that the right-hand side is a function of \( q \) — which we can choose freely — minimizing \( KL(q|\pi) \) with respect to \( q \) aligns with maximizing the right-hand side of Equation \eqref{ineqpsi}. This strategy aids in pinpointing the optimal likelihood approximation.

To address the optimization challenge, the variational approximation method constrains the approximate likelihood of observing network \( g \) to a factorized form:
\begin{equation*}
  q(g)=\prod_{i=1}^n \prod_{j\neq i} \mu_{ij}^{g_{ij}} (1-\mu_{ij})^{1-g_{ij}},
\end{equation*}
where \( \mu_{ij} \) represents the probability \( \Pr_{q_n}(g_{ij}=1) \). This formulation is referred to as the mean field approximation for the likelihood of \( g \), signifying that the likelihood \( q \) is represented by the product of the likelihoods for individual links.

Although this configuration lacks a microeconomic foundation, \( q(g) \) covers the entire likelihood space for \( g \) in \( \mathcal G \), providing justification for this approximation form. Direct algebraic computation reveals the entropy of \( q \) as:
\begin{equation*}
  \frac{1}{n^2}\mathcal{H}(q) = -\frac{1}{2n^2} \sum_{i=1}^n \sum_{j\neq i} \left[ \mu_{ij}\log \mu_{ij} + (1-\mu_{ij})\log (1-\mu_{ij}) \right].
\end{equation*}

Referring back to the potential function in Equation \eqref{eq:potentialfunction}, the expected value of the re-scaled potential is:
\begin{align*}
 \mathbb{E}_q[T(g,X;\theta)] &= \sum_{i=1}^n \sum_{j\neq i} \frac{1}{n^2} \mu_{ij}[ u(X_i,X_j;\beta) ] 
           + \frac{1}{n^2}\sum_{i=1}^n\sum_{j\neq i} \mu_{ij}\mu_{ji}\gamma_1 \\
   &+ \frac{1}{n^2} \sum_{i=1}^n\sum_{j\neq i} \sum_{k\neq i,j} \mu_{ij}\mu_{jk}\gamma_2.
\end{align*}

With the constraints on \( q \), our objective is to discover a vector \( \mathbf{\mu} \) for:
\begin{equation*}
    \sup_{\mathbf{\mu}\in [0,1]^{n(n-1)}} \left\{
\mathbb{E}_q[T(g,X;\theta)]+\frac{1}{n^2}\mathcal{H}(q)
    \right\},
\end{equation*}
where \( \mathbf{\mu} \) consists of elements \( \mu_{ij} \) for distinct \( i \) and \( j \). Substituting the optimized estimates of \( \mathbf \mu \) into the aforementioned supremum yields \( \phi^{MF}(X;\theta) \) as an approximation for \( \phi(X;\theta) \). Consequently, \( c^{MF}(X;\theta) = \exp[n^2 \phi^{MF}(X;\theta)] \) serves as an approximation for the normalizing constant.

Theorem 1 from \cite{melezhu2017approximate} asserts that, for any given \( X \) and \( \theta \):
\begin{equation*}
     \phi(X;\theta) - \phi^{MF}(X;\theta) \rightarrow 0
    \quad \text{ as }  n \rightarrow \infty,
\end{equation*}
affirming the validity of the mean field approximation method.

\subsubsection{Implementation}

Once the approximated normalizing constant $c^{MF}(X;\theta)$ is acquired, the Metropolis-Hasting method can be applied to sample pairs $(\theta,g)$ from the posterior distribution given by Equation \eqref{posteriortwovar}. The procedure is outlined as follows:

\begin{algorithm}[Metropolis-Hasting for ARD]\label{algo:ard}
Iterate the steps for a total of \( T \) times. At the \( t \)-th iteration with current parameter value \( \theta \) and network \( g \):

1. Propose a new parameter value \( \theta' \) using the proposal distribution \( q_{\theta}(\cdot | \theta) \). Generate a new network \( g' \) via the network formulation algorithm, given \( \theta \).

2. Accept the pair \( (\theta',g') \) with probability:
\begin{align*}
&\alpha(\theta,g,\theta',g'|\psi_0,X) \\
&= \min\left\{ 1, 
\frac{p_c(\theta',g')}{p_c(\theta,g)}
\frac{\exp[\bar Q(g',X;\theta')] / c^{MF}(X;\theta')}  {\exp[\bar Q(g,X;\theta)] / c^{MF}(X;\theta)}
\frac{\mathds{1}\left\{ \psi(g',X)=\psi_0 \right\}}{\mathds{1}\left\{ \psi(g,X)=\psi_0 \right\}} \cdot \right. \\
&\quad \left. \frac{q_{\theta}(\theta|\theta')}{q_{\theta}(\theta'|\theta)}
\frac{\exp[\bar Q(g,X;\theta')] / c^{MF}(X;\theta')}  {\exp[\bar Q(g',X;\theta)] / c^{MF}(X;\theta)}
\right\}.
\end{align*}
\end{algorithm}

Note that in Step 1, \( g' \) is simulated using \( \theta \) rather than \( \theta' \) since it serves as the proposal. As such, the proposal distribution for \( g' \) corresponds to the stationary distribution of networks given \( \theta \).

Algorithm \ref{algo:ard} still demands the generation of a network \( g' \) from \( \theta \). It also incorporates an indicator to verify if \( g' \) aligns with the ARD statistic \( \psi_0 \). Another such indicator function appears in the denominator as per the Metropolis-Hasting formulation. This is not an issue, as during repeated algorithm iterations, only networks consistent with the ARD statistic are accepted, ensuring the denominator's indicator is always unity. Importantly, if the algorithm does not initiate from a \( g \) where \( \psi(g,X)=\psi_0 \), the acceptance ratio becomes \( \min\{ 1,\infty \} \), implying the proposal is always accepted until a \( g \) consistent with the ARD emerges.

The application of the Metropolis-Hasting framework ensures the algorithm's convergence. Substituting $c(X;\theta)$ with $c^{MF}(X;\theta)$ in Equation \eqref{stationarydistri}, we derive 
\begin{equation*}
  \pi^{MF}(g|X;\theta)=\frac{\exp[\bar Q(g,X;\theta)]}{ c^{MF}(X;\theta) }
\end{equation*}
By introducing $\pi^{MF}(g|X;\theta)$ into Equation \eqref{posteriortwovar} in place of $\pi(g|X;\theta)$, we obtain the mean field approximated posterior distribution 
\begin{equation*} 
  p^{MF}(\theta,g|\psi_0,X)=
  \frac{p_c(\theta,g)  \mathds{1}\left\{ \psi(g,X)=\psi_0 \right\} \pi^{MF}(g|X;\theta) }
  {\mathcal{Z}^{MF}(\psi_0,X)},
\end{equation*}
where the normalizing term is defined as: 
\begin{equation*}
  \mathcal{Z}^{MF}(\psi_0,X)= \int_{\Theta} \sum_{\omega\in\mathcal{G}} p_c(\vartheta,\omega) \mathds{1}\left\{ \psi(\omega,X)=\psi_0 \right\} \pi^{MF}(\omega|X;\vartheta)  d\vartheta.
\end{equation*}
The convergence property of this methodology is articulated in the following theorem:

\begin{theorem}\label{convergeofard}
Given adequate iterations, the distribution of \( (\theta,g) \) within Algorithm \ref{algo:ard} will converge to the posterior distribution $p^{MF}(\theta,g|\psi_0,X)$.
\end{theorem}

In conjunction with Theorem 1 from \cite{melezhu2017approximate}, which underscores the proximity between the normalizing constant and its mean field approximation, Theorem \ref{convergeofard} affirms that the algorithm reliably converges to the observed ARD statistics' pertinent posterior distribution, albeit through a mean field approximation.

From the sampled pairs \( (\theta,g) \), by discarding \( g \), the residual \( \theta \) constitutes a sample from \( \theta \)'s marginal distribution. Moreover, the choice of the ARD statistic does not influence convergence, though it might distort the underlying network's posterior distribution, as discussed above.

An inherent challenge in the practical implementation of the algorithm arises when proposing a network \(g'\) that satisfies \(\psi(g',X)=\psi_0\). This satisfaction can be elusive due to the strictness of \(\psi\) which might allow only for a narrow range of \(g\) compared to the vast set \(\mathcal{G}\). When \(\psi\) is particularly restrictive, it becomes hard to propose a suitable \(g\) that meets the requirements, thereby causing the algorithm to potentially converge slowly.

To circumvent this hurdle, I introduce a relaxation mechanism. Instead of strictly requiring $\psi(g',X)=\psi_0$, I replace the indicator function $\mathds{1}\left\{ \psi(g',X)=\psi_0 \right\}$ with $\mathds{1}\left\{ \|\psi(g',X)-\psi_0\| \leq \delta \right\}$, where $\| \cdot \|$ represents the norm for $\psi$, and $\delta$ is a predetermined tolerance level. This replacement results in the algorithm converging to the posterior distribution defined as:
\begin{equation} \label{posteriorleq}
  p^{MF}_{\delta}(\theta,g|\psi_0,X) = 
  \frac{p_c(\theta,g) \frac{\exp[\bar Q(g,X;\theta)]}{c^{MF}(X;\theta)} \mathds{1}\left\{ \|\psi(g,X)-\psi_0\| \leq \delta \right\} }
  {\mathcal{Z}^{MF}_{\delta}(\psi_0,X)},
\end{equation}
with the denominator given by:
\begin{equation*}
  \mathcal{Z}^{MF}_{\delta}(\psi_0,X) = \int_{\Theta} \sum_{\omega\in\mathcal{G}} p_c(\vartheta,\omega) \frac{\exp[\bar Q(\omega,X;\vartheta)]}{c^{MF}(\vartheta,X)} \mathds{1}\left\{ \|\psi(g,X)-\psi_0\| \leq \delta \right\}  d\vartheta.
\end{equation*}
The distribution from Equation \eqref{posteriorleq} can be seen as a relaxed or smoothed version of the one in \eqref{posteriortwovar}. By introducing the $\delta$ parameter, the model allows for minor variations, thereby adding robustness. Notably, as $\delta$ approaches zero, Equation \eqref{posteriorleq} reverts to the stricter form of \eqref{posteriortwovar}.

For practical purposes, one can adjust $\delta$ adaptively within the algorithm to control the acceptance rate of proposed $(\theta,g)$ pairs. A persistently high acceptance rate would signal the need to reduce $\delta$ for greater precision, while a low rate would necessitate its increase to avoid undue stringency. Striving to maintain an optimal acceptance rate, around 5\%, guarantees both algorithmic efficiency and faithfulness to observed data.

This adaptive stance addresses two pivotal challenges in Bayesian computations: ensuring efficient sampling and accommodating the intricacies of real-world data. By striking this balance, the algorithm's utility is preserved both computationally and empirically.

\section{Monte Carlo Simulation Analysis}

In this analysis, I delve into the evaluation of Aggregate Relational Data for its efficacy in estimating network formation models. The experiment proceeds in a sequential three-stage process outlined as follows:

Initially, I simulate networks, leveraging a well-defined network formation model. It generates a controlled setting that mirrors the real-world network dynamics, yet under a regime where the true parameters are known. To ensure robustness and minimize the influence of randomness inherent to network formation, I sample multiple networks from their stationary distribution. 

Following the network simulation, the next phase involves the generation of Aggregate Relational Data  for each simulated network. This step translates the complex interconnections and interactions within the network into a simplified and aggregated form that retains essential relational information. 

Advancing to the estimation phase, I sample from the posterior distribution of the parameters, utilizing the Bayesian inference techniques outlined earlier. The estimation process culminates in the construction of credible sets for the parameters, derived from the ARD-based estimations. These sets are then compared against the true parameters used in the network simulations. Through this empirical  analysis, I aim to examine the potential of ARD as a tool for understanding and modeling complex networks.

\subsection{Design 1: No Social Interaction}\label{subsec:reducedcase}

The initial phase of my investigation delves into a simplified model of network formation, focusing exclusively on bilateral relationships. This model uses a utility function defined as:
\begin{equation*}
U_i(g,X,\varepsilon;\theta)=
\sum_{j\neq i} g_{ij} ( 5+ \theta|age_i-age_j| +\varepsilon_{ij} ),
\end{equation*}
where $\theta$ reflects the impact of age differences on the likelihood of forming a link. The model excludes considerations of reciprocity, indirect connections, or popularity effects. This approach allows for the validation of the fundamental principles of the method in a controlled setting before addressing more complex network dynamics.

For simulation purposes, I have set $\theta = -1$, indicating that the formation of relationships is negatively influenced by age differences. Then, I  compile Aggregate Relational Data for analysis, which forms the basis of my Monte Carlo Simulation.

The foundation of this study relies on a dataset from detailed surveys carried out in a village in Karnataka, southern India, in 2006. This dataset, analyzed in-depth by \cite{jackson2012social}, provides a detailed view of the village's social structure, highlighting the varied characteristics of individuals in the network. The data specifically comes from Village 47 in their broader dataset, offering detailed insights into individual characteristics for simulation purposes.

\newcounter{ARDcounter}
\setcounter{ARDcounter}{1}

The process of calibrating Aggregate Relational Data  involves analyzing responses to carefully designed questions asked of each individual $i$ in the network. This ARD captures key aspects of the network's relational dynamics, which are crucial for my analysis. The questions, listed below, are designed to convert important elements of the network structure into measurable data, recorded as $\psi_i$ for each participant:
\begin{enumerate}
  \item The total number of inbound links, $\psi_{i\arabic{ARDcounter}}$. \stepcounter{ARDcounter}
  \item The total number of outbound links, $\psi_{i\arabic{ARDcounter}}$. \stepcounter{ARDcounter}
  \item The count of inbound links from individuals within a 5-year age difference, $\psi_{i\arabic{ARDcounter}}$. \stepcounter{ARDcounter}
  \item The count of outbound links to individuals within a 5-year age difference, $\psi_{i\arabic{ARDcounter}}$. \stepcounter{ARDcounter}
  \item The number of inbound links from individuals aged less than 24 years, $\psi_{i\arabic{ARDcounter}}$. \stepcounter{ARDcounter}
  \item The number of outbound links to individuals aged less than 24 years, $\psi_{i\arabic{ARDcounter}}$. \stepcounter{ARDcounter}
  \item The count of inbound links from individuals aged between 25 to 44 years, $\psi_{i\arabic{ARDcounter}}$. \stepcounter{ARDcounter}
  \item The count of outbound links to individuals aged between 25 to 44 years, $\psi_{i\arabic{ARDcounter}}$. \stepcounter{ARDcounter}
  \item The number of inbound links from individuals aged more than 45 years, $\psi_{i\arabic{ARDcounter}}$. \stepcounter{ARDcounter}
  \item The number of outbound links to individuals aged more than 45 years, $\psi_{i\arabic{ARDcounter}}$. \stepcounter{ARDcounter}
\end{enumerate}
Structured as $\psi_i=[\psi_{i1},...,\psi_{i\arabic{ARDcounter}}]'$ for each individual and collectively as $\psi=[\psi_1',...,\psi_n']'$ for the entire network, the ARD provides a detailed representation of the network's overall relational attributes, setting the foundation for further analysis and estimation.

The simulation starts with an empty network and evolvs through 100,000 iterations. During these iterations, $\theta$ is adjusted across 30,000 iterations to explore the parameter space thoroughly. I adjust the tolerance level, $\delta$, twenty times to balance convergence speed, accuracy, and computational efficiency, ensuring it matches the expected network behavior closely. The evolution of $\theta$ is divided into 150 rounds, each of which contains 200 draws of $\theta$. The adjustment of $\delta$ follows the rule:
\begin{itemize}
    \item If the ratio of accepted draws is more than 50\%, $\delta$ is adjusted to 0.98 times its value.
    \item If the ratio of accepted draws is more than 18\% and less than 50\%, $\delta$ is adjusted to 0.99 times its value.
    \item If the ratio of accepted draws is more than 1\% and less than 18\%, $\delta$ is not adjusted in this round.
    \item If the ratio of accepted draws is more than 0.3\% and less than 1\%, $\delta$ is adjusted to 1.005 times its value.
    \item If the ratio of accepted draws is less than 0.3\%, $\delta$ is adjusted to 1.01 times its value.
\end{itemize}

The changes in $\theta$ during the simulation are shown in Figure \ref{simplesttrend}. This figure shows $\theta$ quickly stabilizes around the true value of $-1$, with fluctuations remaining within the $(-1.2, -1)$ range, tightly encompassing the true value with minimal variation.

\begin{figure}[htbp]
    \centering
    \includegraphics[width=0.8\linewidth]{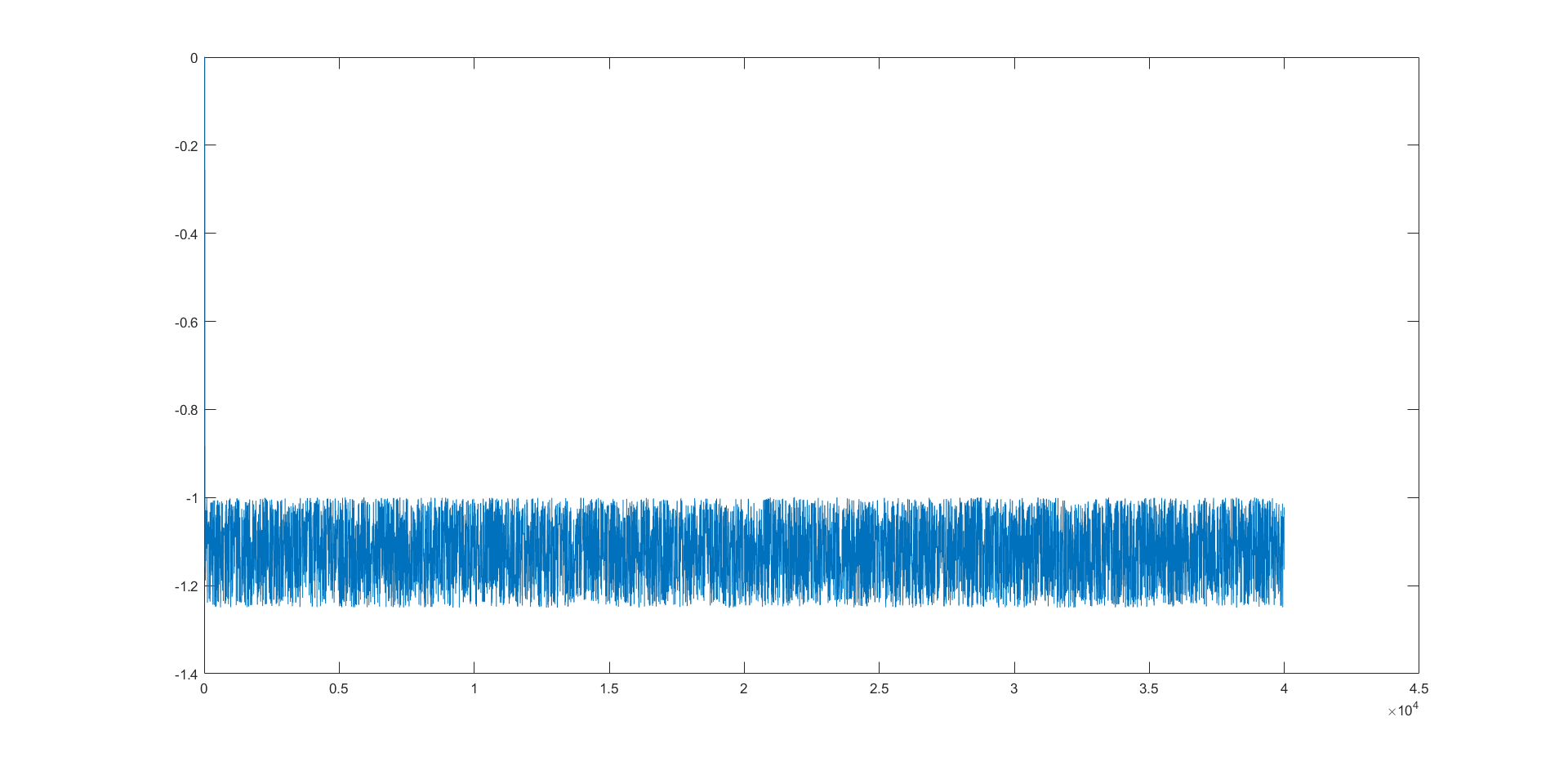}
\caption{Trend of $\theta$ without Social Interaction}
\label{simplesttrend}
\end{figure}

This slight deviation is due to the nature of ARD. By compressing network data into a set of measurable responses, ARD naturally results in some information loss. This effect is shown visually by expanding a precise point estimate into a small interval, indicating the balance between data simplification and accuracy.

The results of this simulation support the effectiveness of the proposed method for ARD-based scenarios, showing that $\theta$ converges to the posterior distribution when using ARD. This supports the convergence efficiency of Algorithm \ref{algo:ard}, confirming the value of this approach in applying ARD to network model estimations.

\subsection{Design 2: With Social Interaction}

In this part, I incorporate the concept of reciprocal utility to enrich the utility function. This enhancement enables my model to more accurately capture the social interactions within networks:
\[
U_i(g,X,\varepsilon;\theta)=
  \sum_{j\neq i} g_{ij} \left(\beta_0 +  \frac{\beta_1|age_i-age_j|}{20} + \varepsilon_{ij}\right)
    + \sum_{j\neq i} g_{ij}g_{ji} \gamma_1,
\]
where $\theta = (\beta_0, \beta_1, \gamma_1)'$ is selected to capture the complex dynamics of link formation, especially focusing on age differences and reciprocity in relationships. I normalize the age difference by dividing it by 20. This adjustment helps to ensure that the impact of age differences on the tendency to form links is represented on a practical and understandable scale.

In my Monte Carlo simulations, I use a setting that assumes minimal influence from mutual utility ($\gamma_1$), with $\beta_0$ set at 1. This configuration, represented by $\theta = (1, -1, 0.1)$, establishes a basis for examining the model's behavior under simplified assumptions about the influence of mutual utility.

I examine two distinct sets of ARDs. The first set comprises the following statistics:
\begin{enumerate}
  \item The total number of inbound links.
  \item The total number of outbound links.
  \item The count of inbound links from individuals within a 5-year age difference.
  \item The count of outbound links to individuals within a 5-year age difference.
  \item The count of inbound links from individuals with a 5-year to 10-year age difference.
  \item The count of outbound links to individuals with a 5-year to 10-year age difference.
  \item The count of inbound links from individuals with a more than 10-year age difference.
  \item The count of outbound links to individuals with a more than 10-year age difference.
\end{enumerate}

I use the same pattern to adapt $\delta$ as in Section \ref{subsec:reducedcase} with 30,000 iterations of $\theta$. 
After four repeated simulations, the average 90\% credible intervals for $\hat\beta_1$ and $\hat\gamma_1$ are $(-1.5099, 2.1415)$ and $(-1.4140, 2.2936)$ respectively, covering the true parameters but showing a wide range of possible parameter values. 
Figure \ref{fig:structrendsimpleard} shows a representative simulation with changes in $\theta$, where the rounds which has a extreme $\delta$ value and incurs unreasonable dynamics of $\theta$ are dropped. 
This range indicates the complexity that arises from allowing two parameters to vary simultaneously, which considerably affect the network's stationary distribution along with the randomness added by the $\varepsilon$ term. This situation highlights the complex trade-off between making the model easier to interpret through flexible parameters and the difficulty in dealing with the stochastic aspects of network dynamics.

\begin{figure}[htbp]
    \centering
    \includegraphics[width=0.8\linewidth]{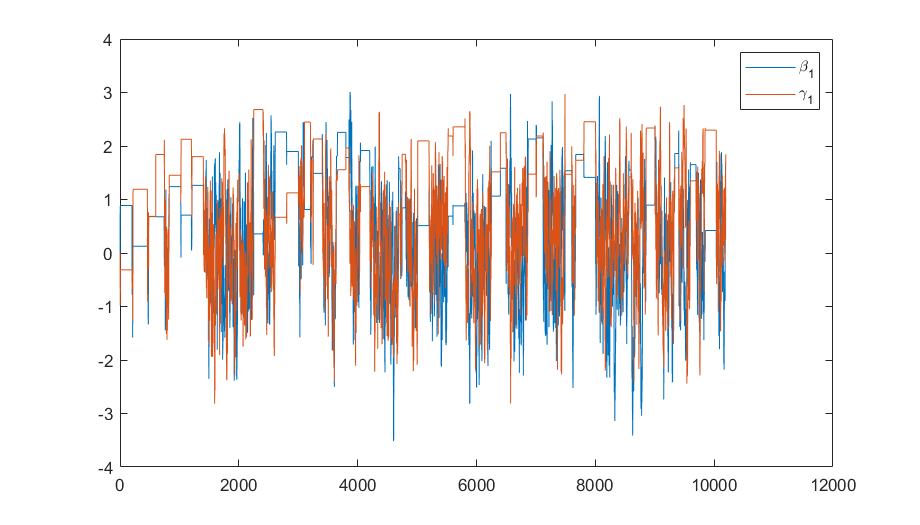}
\caption{Trend of $\theta$ with Social Interaction}
\label{fig:structrendsimpleard}
\end{figure}

Adding reciprocal link terms introduces more complexity to the model, affecting the network's stationary distribution based on $\theta$. Without link externalities, this distribution could be more simply approximated by a multinomial distribution. However, introducing such externalities complicates the model, necessitating more complex approximations in the variational approach. This change might make the results less aligned with the intended distribution, emphasizing the intricate challenges brought by broader structural elements of the model.

Continuing my investigation into the influence of ARD selection, I explore a second set of ARD statistics that includes questions on a finer division of age difference. This augmented set is comprised of the following:
\begin{enumerate}
  \item The total number of inbound links.
  \item The total number of outbound links.
  \item The count of inbound links from individuals within a 2-year age difference.
  \item The count of outbound links to individuals within a 2-year age difference.
  \item The count of inbound links from individuals with a 2-year to 5-year age difference.
  \item The count of outbound links to individuals with a 2-year to 5-year age difference.
  \item The count of inbound links from individuals with a 5-year to 8-year age difference.
  \item The count of outbound links to individuals with a 5-year to 8-year age difference.
  \item The count of inbound links from individuals with a 8-year to 12-year age difference.
  \item The count of outbound links to individuals with a 8-year to 12-year age difference.
  \item The count of inbound links from individuals with a 12-year to 17-year age difference.
  \item The count of outbound links to individuals with a 12-year to 17-year age difference.
  \item The count of inbound links from individuals with a 17-year to 24-year age difference.
  \item The count of outbound links to individuals with a 17-year to 24-year age difference.
  \item The count of inbound links from individuals with a more than 24-year age difference.
  \item The count of outbound links to individuals with a more than 24-year age difference.
\end{enumerate}

This augmented set of ARD, by finer queries on age differences, aims to enrich the model's input data. With repeated simulations, the average 90\% credible intervals obtained for $\hat\beta_1$ and $\hat\gamma_1$ are $(-1.3441, 2.1515)$ and $(-1.1273, 2.3580)$, respectively. The results are summarized in Table \ref{tab:summary_of_intervals}. 
Figure \ref{fig:structrendsimpleard2} shows a representative simulation with changes in $\theta$, excluding the rounds having an extreme $\delta$ value and unreasonable dynamics of $\theta$. 

\begin{table}
\centering
\begin{tabular}{@{}lcc@{}}
\toprule
Model & $\hat\beta_1$ 90\% CI & $\hat\gamma_1$ 90\% CI \\ \midrule
Benchmark ARD   & $(-1.5099, 2.1415)$ & $(-1.4140, 2.2936)$ \\
Augmented ARD   & $(-1.3441, 2.1515)$ & $(-1.1273, 2.3580)$ \\
\bottomrule
\end{tabular}
\caption{Summary of 90\% Credible Intervals for $\hat\beta_1$ and $\hat\gamma_1$}
\label{tab:summary_of_intervals}
\end{table}

\begin{figure}[htbp]
    \centering
    \includegraphics[width=0.8\linewidth]{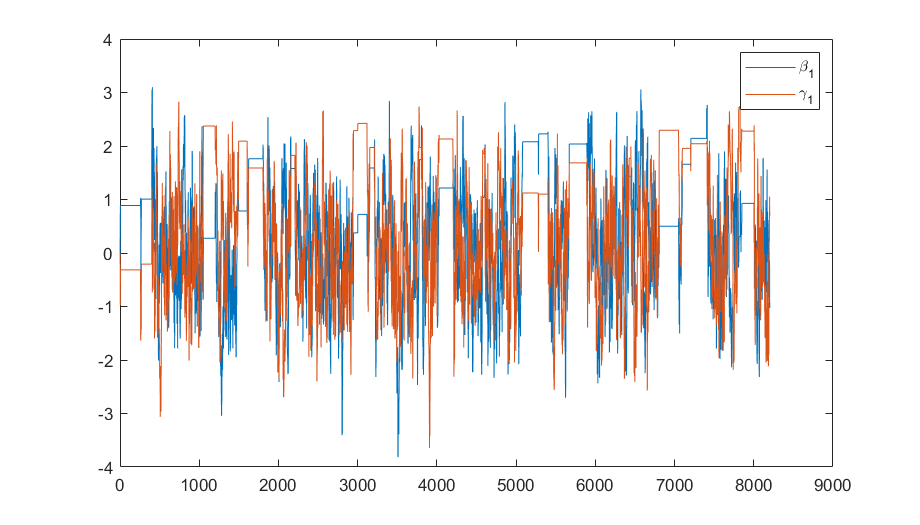}
\caption{Trend of $\theta$ in Structure Case with Augmented ARD}
\label{fig:structrendsimpleard2}
\end{figure}

These results suggest that augmenting the ARD with additional statistics can refine the parameter estimates to a certain extent. However, the efficacy of such refinements largely depends on the relevance and nature of the included ARD statistics, as previously discussed in Section \ref{sec:appendix:ard}. In scenarios where the ARD primarily captures age differences, introducing a finer division may not significantly alter the informational landscape. Consequently, the observed adjustments in the estimates of $\theta$ are modest, underscoring the impact of ARD composition on model estimations.

In summary, using Aggregate Relational Data  in my analysis yielded insightful yet complex results. Although the credible intervals consistently included the true parameter values, obtaining precise estimates proved difficult. These outcomes stem from several critical factors:

\textbf{Convergence and Model Complexity:} Integrating ARD can introduce convergence challenges in the simulation algorithm, reflecting a balance between data simplification and estimation accuracy. This situation highlights the necessity for further refinements in algorithms and adjustments in models to more effectively capture the complex dynamics of network structures.

\textbf{Computational Tools and Framework Limitations:} Utilizing the R package \texttt{netnew} for network simulations added complexity to my analysis. While this tool was crucial for conducting my research, its inherent convergence behavior and interaction with model parameters might require careful evaluation. This observation is not a critique but an acknowledgment of the ongoing need to develop and assess such computational tools.

Despite these challenges, the insights from my research reveal the complexities involved in using ARD to estimate complex network formation models. This sets the foundation for future research, including refining simulation algorithms, exploring alternative computational frameworks, and conducting thorough evaluations of model convergence to achieve more accurate estimations.

In aiming for a balanced discussion, I recognize the contributions of existing tools like the \texttt{netnew} package while pointing out areas for enhancement. This approach fosters a constructive dialogue aimed at advancing network analysis methodologies, promoting both methodological development and practical application.

\section{Conclusion}\label{section:conclusion}

In this study, I have extended beyond traditional methods in network data analysis to examine the use of Aggregate Relational Data in estimating network structures. This approach, enhanced by variational approximation methods, tackles the challenges posed by the intractability of the normalizing constant—a long-standing obstacle in complex network analysis. My methodology not only makes the computation process simpler but also offers a cost-effective solution for researchers with data collection constraints, thereby broadening the scope for network analysis in various domains.

My results highlight the utility of ARD in providing insightful approximations of network dynamics, even with limited or impractical access to detailed network data. By selecting and using ARD statistics effectively, I show that it is feasible to obtain meaningful insights into network formation and evolution. However, it is important to recognize the limitations of this approach, particularly its dependence on aggregate data which may not fully capture the complexity of individual interactions.

Looking forward, there are ample opportunities for advancing ARD-based methodologies. A promising direction is to improve the accuracy and dependability of parameter estimation through enhanced statistical methods and algorithmic improvements. This could lead to more detailed models that better reflect the subtleties of network data.

Additionally, exploring identified sets using ARD poses a fascinating challenge for future research. Comparing these sets with credible sets from Bayesian estimation could provide deeper understanding of ARD’s capabilities and limitations in network analysis. Such comparative studies could foster new methods that balance the theoretical ideal of complete network observation with the practical limits of data collection.

Moreover, applying ARD to a wider array of network datasets is worthwhile. Investigating various datasets, particularly those with different complexities and detail levels, might illuminate the generalizability and adaptability of ARD-based approaches. This exploration could identify key factors affecting ARD’s effectiveness in different network scenarios.

In conclusion, this study represents a significant advancement in using Aggregate Relational Data for network analysis, opening up numerous research possibilities. The ongoing effort to fully exploit ARD in deciphering complex social and economic networks continues. Future work should strive to expand what is feasible in network analysis and enhance our understanding of the interconnected world we live in.

\bibliographystyle{ecta}
\bibliography{citation.bib}

\appendix

\section{Proof}\label{appendixproof}

\begin{proof}[Proof of Proposition \ref{prop:suffiexamp}]

\textbf{Step 1: Initial Deductions}  
From the answer of A to Question 3, we can determine the value of $g_{AB}$. Similarly, by examining the answer of B to Question 3, the value of $g_{BA}$ becomes apparent. Now, by comparing the responses of C to Questions 1 and 3, we can derive the value of $g_{CD}$. Following the same methodology, by contrasting D's answers to these two questions, we discern the value of $g_{DC}$. 

\textbf{Step 2: Deciphering Links}  
We delve deeper by evaluating the situations of A's links to C and D:
\begin{itemize}
    \item \textbf{A's Both or No Links:} Analyzing the response of A to Question 1, if he indicates either no links or links to both C and D, that is, $g_{AC}=0, g_{AD}=0$ or $g_{AC}=1, g_{AD}=1$, then C's response to Question 2 unveils $g_{BC}$, and similarly, D's response to the same question gives us $g_{BD}$. 
    Here, the scenario bifurcates by the situations of C's links to A and B:
      \begin{itemize}
          \item \textbf{C's Both or No Links:} If C's response to Question 1 indicates either both or no links to A and B, that is, $g_{CA}=0, g_{CB}=0$ or $g_{CA}=1, g_{CB}=1$, then A's response to Question 2 lets us deduce $g_{DA}$, and B's response to the same question unveils $g_{DB}$. At this point, the network structure can be wholly mapped using ARD.
          \item \textbf{C's Single Link:} If C's response indicates a single link to either A or B, represented as $g_{CA}+g_{CB}=1$, we then entertain two possible scenarios. 
          \begin{itemize}
              \item In the first scenario, we assume $g_{CA}=1$ and consequently $g_{CB}=0$. By leveraging the subsequent answers from A and B to Question 2, we can determine the remaining network links $g_{DA}$ and $g_{DB}$. 
              \item In the alternate scenario where $g_{CB}=1$ and $g_{CA}=0$, we can similarly discern the rest of the network links from the answers of A and B to Question 2. 
          \end{itemize}
           Regardless of which assumption is made, two potential network structures emerge from the ARD. Turning our attention back to the utility function, it is evident that both $g_{AC}$ and $g_{AD}$ yield identical utility to A but not to others; similarly, $g_{CA}$ and $g_{CB}$ offer equal utility to C but none to others. Given this, we can conclude that both potential networks — distinguished only by the links between pairs A, B and C, D — are equally probable in a stationary distribution. This leads to an observational probability of 50\% for either network based on the ARD.
      \end{itemize}
    \item \textbf{A's Singular Link:} If A's response indicates a link to either C or D, that is, $g_{AC}+g_{AD}=1$, irrespective of the exact configuration of $g_{AC}=1$, $g_{AD}=0$ or $g_{AC}=0$, $g_{AD}=1$, we proceed with the established analytical framework by the situations of C's links to A and B:
      \begin{itemize}
          \item \textbf{C's Both or No Links:} If C indicates having links to both or none between A and B, an analogous analysis as above considering the utility function presents two feasible networks, each holding an observational probability of 50\%.
          \item \textbf{C's Singular Link:} In case C hints at having just one link between A and B, that is, $g_{CA}+g_{CB}=1$, similar analysis reveals that there are potentially four network configurations. These networks mainly differ based on links between pairs A, B and C, D. Relying upon the utility function, where no individual differentiates between links to any member of the opposite pair, the probability of observing any one of the four configurations stands at 25\%.
      \end{itemize}
\end{itemize}

\textbf{Step 3: Final Observations}  
In light of the above, we recognize that $\Pr[g|\psi(g,X);\theta]$ remains invariant with respect to $\theta$ across all cases. This brings us to the deduction that the ARD statistics hold sufficiency for $\theta$.
\end{proof}

\begin{proof}[Proof of Theorem \ref{sufficientaux}]
When $\psi$ is sufficient for $\theta$, $\Pr[g|\psi(g,X);\theta]$ is not revelent of $\theta$, then
the posterior distribution of observing $\psi(g_0,X)$ is
\begin{eqnarray*}
  p(\theta|\psi(g_0,X),X) &=&
  \frac{p_0(\theta) L(\psi(g_0,X);\theta) }    { \int_{\Theta} p_0(\vartheta) L(\psi(g_0,X);\vartheta) d\vartheta }  \\
  &=& \frac{p_0(\theta) P[g_0|\psi(g_0,X)] L(g_0;\theta) }    { \int_{\Theta} p_0(\vartheta) P[g_0|\psi(g_0,X)] L(g_0;\vartheta) d\vartheta } \\
   &=& p(\theta|g_0,X)
\end{eqnarray*}
which coincides with the posterior distribution of observing $g_0$. 
\end{proof}

\begin{proof}[Proof of Theorem \ref{convergeofard}]
Algorithm \ref{algo:ard} follows standard Metropolis-Hasting framework. We directly verify the detailed balanced condition with the posterior distribution in Equation \eqref{posteriortwovar}.
\begin{eqnarray*}
    & & p(\theta,g|\psi_0,X) \Pr[\theta',g'|\theta,g,\psi_0,X]  \\
    &=& p(\theta,g|\psi_0,X) q_\theta(\theta'|\theta) \pi(g'|X;\theta) \alpha(\theta,g,\theta',g'|\psi_0,X) \\
    &=& \frac{p_c(\theta,g)  \mathds{1}\left\{ \psi(g,X)=\psi_0 \right\} \exp[Q(g,X;\theta)] / c(X;\theta) }{ \mathcal{Z}(\psi_0,X)} q_\theta(\theta'|\theta)  \exp[Q(g',X;\theta)] / c(X;\theta) \cdot   \\
    & &
    \min\left\{ 1,
  \frac{p_c(\theta',g')}{p_c(\theta,g)}
  \frac{\exp[Q(g',X;\theta')] / c(X;\theta')}  {\exp[Q(g,X;\theta)] / c(X;\theta)}
  \frac{\mathds{1}\left\{ \psi(g',X)=\psi_0 \right\}}{\mathds{1}\left\{ \psi(g,X)=\psi_0 \right\}}
  \frac{q_{\theta}(\theta|\theta')}{q_{\theta}(\theta'|\theta)}
  \frac{\exp[Q(g,X;\theta')] / c(X;\theta')}  {\exp[Q(g',X;\theta)] / c(X;\theta)}
  \right\}    \\
   &=& 
   \frac{1 }{ \mathcal{Z}(\psi_0,X)} \cdot  \\
   & &   \min\left\{ 
   p_c(\theta,g) q_\theta(\theta'|\theta)  \mathds{1}\left\{ \psi(g,X)=\psi_0 \right\} \exp[Q(g,X;\theta)] / c(X;\theta)  \exp[Q(g',X;\theta)] / c(X;\theta), 
   \right.        \\
   & &  \left.
   p_c(\theta',g')  q_{\theta}(\theta|\theta') \mathds{1}\left\{ \psi(g',X)=\psi_0 \right\}  \exp[Q(g',X;\theta')] / c(X;\theta')   \exp[Q(g,X;\theta')] / c(X;\theta')
   \right\} \\
   &=& \min\left\{
 \frac{p_c(\theta,g)}{p_c(\theta',g')}
\frac{q_\theta(\theta'|\theta)}{q_{\theta}(\theta|\theta')}
\frac{\mathds{1}\left\{ \psi(g,X)=\psi_0 \right\}}{\mathds{1}\left\{ \psi(g',X)=\psi_0 \right\}}
\frac{\exp[Q(g,X;\theta)] / c(X;\theta)}{\exp[Q(g',X;\theta')] / c(X;\theta')}
\frac{\exp[Q(g',X;\theta)] / c(X;\theta)}{\exp[Q(g,X;\theta')] / c(X;\theta')},
  1   \right\} \cdot  \\
  & &   q_\theta(\theta|\theta')  \exp[Q(g,X;\theta')] / c(X;\theta') \cdot
  \frac{p_c(\theta',g')  \mathds{1}\left\{ \psi(g',X)=\psi_0 \right\} \exp[Q(g',X;\theta')] / c(X;\theta') }{ \mathcal{Z}(\psi_0,X)}   \\
  &=&  q_\theta(\theta|\theta') \pi(g|X;\theta')   \alpha(\theta',g',\theta,g|\psi_0,X)
  p(\theta',g'|\psi_0,X)  \\
  &=& \Pr[\theta,g|\theta',g',\psi_0,X] p(\theta',g'|\psi_0,X) 
\end{eqnarray*}

This indicates that the detailed balance condition is satisfied. Therefore, according to the nature of Metropolis-Hasting, Algorithm \ref{algo:ard} will produce a sequence of $(\theta,g)$ that converges to the posterior distribution in Equation \eqref{posteriortwovar}. 
\end{proof}

\end{document}